\documentclass[11pt]{article}

\usepackage{fullpage}
\usepackage{amsmath}
\usepackage{graphicx}
\usepackage{rotating}
\usepackage{amssymb}
\usepackage{color}
\usepackage{amsmath}
\usepackage{amsthm}
\usepackage{mathtools}
\usepackage{cite}

\newtheorem{theorem}{Theorem}
\newtheorem{lemma}[theorem]{Lemma}

\newtheorem{conjecture}[theorem]{Conjecture}

\allowdisplaybreaks

\title{In pursuit of the dynamic optimality conjecture}

\author{
John Iacono\thanks{Research supported by NSF grants CCF-1018370 and 0430849.}
}

\date{
  Polytechnic Institute of New York University, Brooklyn, New York, USA 
  }

\begin{document}

\maketitle

\begin{abstract}
In 1985, Sleator and Tarjan introduced the splay tree, a self-adjusting binary search tree algorithm. Splay trees were conjectured to perform within a constant factor as any offline rotation-based search tree algorithm on every sufficiently long sequence---any binary search tree algorithm that has this property is said to be dynamically optimal. However, currently neither splay trees nor any other tree algorithm is known to be dynamically optimal.  Here we survey the progress that has been made in the almost thirty years since the conjecture was first formulated, and present a binary search tree algorithm that is dynamically optimal if any binary search tree algorithm is dynamically optimal.
\end{abstract}

\bibliographystyle{alpha}

\newcommand{\OPT}{\mathrm{OPT}}
\newcommand{\IRB}{\mathrm{IRB}}
\newcommand{\ALT}{\mathrm{ALT}}

\section{Introduction}

A \emph{binary search tree (BST)} is a classic structure of computer science. A binary search tree stores a totally ordered set of data and supports the operations of insert, delete, and predecessor queries. Here we focus only on predecessor queries which we call \emph{searches}.
Since there are no insertions and deletions, we can assume the tree contains the integers from 1 to $n$.

To execute each search in the BST model, there is a single pointer that starts at the root of the tree, and at unit cost can move the pointer to the left child, right child, parent, or perform a rotation with the parent (we call these \emph{BST unit-cost primitives}). In order to properly execute the search it is required that the result of the search be touched by the pointer in the course of executing the search. This model was formalized in \cite{DBLP:journals/siamcomp/Wilber89}.

We consider search sequences $X$ of length $m$: $X=x_1, x_2, \ldots x_m$. To avoid issues with small sequences and the initial state of the tree, we assume $m$ is sufficiently long (often only $m=\Omega(n)$ is needed) and that the tree is in some canonical initial state. 
A BST-model algorithm is simply a way of choosing a sequence of BST unit cost primitives to execute each search. A BST-model algorithm is \emph{online} if its choice of BST unit cost primitives to execute search $x_i$ is a function of $x_1, \ldots x_i$. The online BST model is still very permissive, as only BST-model unit cost operations are counted, and unlimited computation could be done to determine these operations. What is normally thought of as a BST is an online BST model algorithm that can be implemented on a BST where $O(\log n)$ bits of data can be augmented on every node, and where unit cost operations are chosen based on the current search, the contents of the node currently pointed to, including any augmented data, and $O(\log n)$ bits of global state. Such a BST algorithm is called a \emph{real-world BST}, a term coined by \cite{DBLP:conf/icalp/BoseCFL12}.
We let $R_A(X)$ denote the cost in the BST model to execute $X$ using some BST-model algorithm $A$.

Let $\OPT(X)$ be the fastest runtime of any BST that can execute $X$; that is $\OPT(X)=\min_A R_A(X)$. Given enough time (i.e.~exponential in $m$), $\OPT(X)$ can be computed exactly, and an offline algorithm $A$ such that $R_A(X)=\OPT(X)$ can be produced.

Splay trees are a BST structure introduced by Sleator and Tarjan \cite{DBLP:journals/jacm/SleatorT85}. They use a simple restructuring heuristic, to move the searched item to the root. This heuristic has the following effect on nodes other than the one searched: if the node $x$ is at depth $d$ and $l$ of the ancestors of $x$ are on the search path, after the search $x$ will be at depth $d+\frac{l}{2}+O(1)$. The work of \cite{DBLP:journals/jal/Subramanian96} explores a class of heuristics that have the same general properties of splay trees. Splay trees have been proven to have a number of amortized bounds on their performance, including such basic facts as $O(\log n)$ amortized runtime per search. However, the focus of this work is on the \emph{dynamic optimality conjecture}:

\begin{conjecture}[Dynamic Optimality Conjecture \cite{DBLP:journals/jacm/SleatorT85}]
$ R_{\mathrm{splay}}(X) = O(\OPT(X))$
\end{conjecture}

We refer to any BST algorithm $A$ such that $R_{A}(X) = O(\OPT(X)+f(n))$ for some $f(n)$ as \emph{dynamically optimal}.
Rather then focus on splay trees themselves, we focus on whether there are any dynamically optimal BSTs. We present several different formulations of this, from weakest to strongest.

\paragraph{Offline optimality.} It is possible to compute in polynomial time (in, say, the RAM model) an algorithm $A$ such that $R_A(X)=O(\OPT(X))$? As we have noted that computing such an algorithm is possible, given enough time, this question concerns only running time, and is the easiest of the questions presented. We believe that computing $\OPT(X)$ exactly is likely to be NP-complete. NP completeness for this very closely related problem was presented in \cite{DBLP:conf/soda/DemaineHIKP09}: instead of a sequence of single searches to be executed on a BST, a sequence of sets of items to be searched are provided and the algorithm can order the searches in each set in whatever manner is beneficial to it. Computing the exact optimal cost for executing such a sequence of sets of searches was shown to be NP-complete. 

\paragraph{Online optimality.} Is there an online BST algorithm $A$ such that $R_A(X)=O(\OPT(X))$?
In this statement of the problem, $A$ could do significant computation in order to determine which BST unit-cost operation to perform at every step, subject only to the requirement that it is online. This conjecture represents the claim that there is no asymptotic difference in power between online and offline algorithms in the BST model. Such equivalence in power between online and offline power is generally not possible in more permissive models, and is typically only found in very strict models such as the optimality of the move-to-front rule for search in a linked list \cite{DBLP:journals/cacm/SleatorT85}. In more permissive models like the RAM, an offline algorithm could fill an array $A$ such that $A[i]=x_i$ and thus could trivially achieve offline performance that an online algorithm could never match.

\paragraph{Online real-world optimality.} The end goal of this line of research is to obtain, not just an online BST algorithm $A$ such that $R_A(X)=O(\OPT(X))$, but one where the runtime in the BST model dominates the runtime and which could be reasonably implemented. The real-world BST model gives a formalization of that goal, and data structures such as splay trees meet the definition of a real-world BST.

Our hope is that solving the offline optimality problem is the bottleneck, and that a solution to that could be transformed into an online real-world algorithm.

We begin this survey by reviewing a geometric view of the problem. We then summarize the results on the lower and upper bounds for $\OPT(X)$. Finally we present a conditional result whereby a concrete online algorithm is presented which is dynamically optimal if any online algorithm is dynamically optimal. Throughout the presentation we try to identify possible avenues for improvement in as well as the challenges of each of the approaches presented.

\section{Geometric view}

We now present a geometric view of a binary search tree algorithm $A$, first introduced in~\cite{DBLP:conf/soda/DemaineHIKP09}. This geometric view was created in the hope that reasoning with this view would be significantly simpler than reasoning with rotation-based trees.

Suppose you have a BST algorithm executing some search sequence $X$. Let the set of points $G_X=\{(x_i,i)\}$ be the \emph{geometric view} of this sequence. Consider the execution of a BST algorithm $A$ on this sequence. Let $t_A(i)$ be all the nodes touched by the pointer in executing $x_i$. Then let $G^A_X=\{(j,i)|j\in t_A(i)\}$; that is, $(i,j)$ is in $G^A_X$ if algorithm $A$ when executing $x_i$ touches $j$. Observe that $G_X \subseteq G^A_X$ for any algorithm $X$, since the item to be searched must be touched in the execution of $X$ by any valid algorithm. This thus gives $G^A_X$ as a plot of everything touched in an execution, seemingly stripping away the specific pointer movements and rotations performed. Also, the runtime of $A$ on $X$ is $O(|G^A_X|)$.

Call two points $p,q$ a set of points $P$ \emph{arborally satisfied (AS)} if there is at least one other point in $P$ in or on the orthogonal rectangle they define. Call 
a set of points $P$ an \emph{arborally satisfied set (ASS)} if all pairs of points $p,q \in P$ that differ in both coordinates are AS.
We have shown that all point sets $G^A_X$ are ASS. Moreover, we have shown that given a ASS point set $P$ where $G_X \subseteq P$, there is a BST algorithm $A$ with runtime $O(|P|)$ such that $G^A_X=P$. 

Thus, the following two problems are equivalent: (1) find an $O(1)$-competitive offline BST to execute $X$, and (2) find an $O(1)$-factor approximation of a minimal ASS superset of $G^A_X$.

One can observe a kind of duality between the time a key is searched and its value.
Suppose $X$ is a permutation of $1..n$. Then the transpose of $G_X$ is also a permutation $X$ and represents some sequence $X^T$ where if $x_i=j$ in $X$ then $x^T_j=i$ in $X^T$. 
As the horizontal and vertical coordinates are treated symmetrically in the definition of ASS, it follows that $\OPT(X)=\OPT(X^T)$. 

Note that this statement of the geometric view is offline. An algorithm for computing an ASS superset $P$ of $G_X$ is said to be online if the points of $P$ are only a function of the points of $G_X$ with the same or smaller $y$-coordinate. Note that an online BST algorithm yields an online ASS superset algorithm directly from the preceding definitions. However, the conversion in the other direction does not preserve online-ness. The method used by Fox \cite{DBLP:conf/wads/Fox11} to turn a particular online ASS superset algorithm into an online BST algorithm could probably be adapted to turn any online ASS superset algorithm into an online BST algorithm.

While computing $\OPT(X)$ offline seems like a clean geometric optimization problem, a solution has remained elusive. The main stumbling block is that it is fairly easy to come up with a minimal superset $P$ of $G_X$ such that all pairs of points in $G_X$ are AS with respect to $P$; the problem is that the points in $P \setminus G_X$ must all also be pairwise AS for the set $P$ to be ASS.

\section{Lower bounds}

One defining feature of this problem is the existence of non-trivial lower bounds on $\OPT(X)$ for particular fixed $X$, as opposed to lower bounds where $X$ is drawn from a distribution. There are several known bounds, we present each of them separately.

\paragraph{Independent rectangle bound \cite{DBLP:conf/soda/DemaineHIKP09}.}
Given a set of rectangles $R$, each of which is defined by two points in $P$, we say $R$ is an independent set of rectangles if no corner of one rectangle lies properly inside of another. Define $\IRB(P)$ to be the size of the maximum set of independent rectangles possible with respect to point set $P$. It has been shown that $\OPT(X)=\Omega(\IRB(P_X))$.

\begin{figure}
\begin{center}
\includegraphics[width=2in]{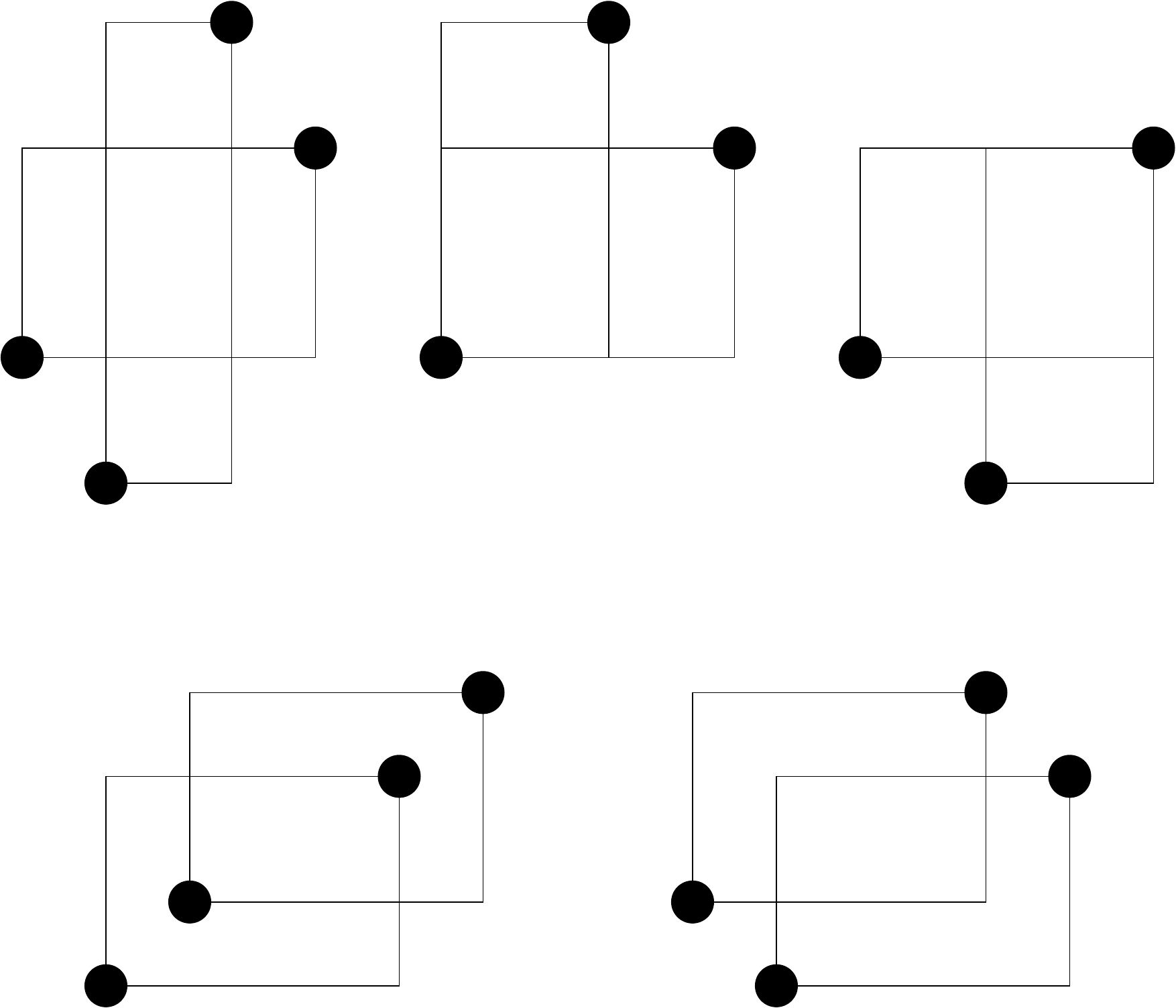}
\end{center}
\caption{Here we consider all combinatorial cases of overlapping $+$-rectangles are independent, and when they are not. The top three pairs of rectangles are independent; the bottom two are not.}
\label{fig:ir}
\end{figure}

Computing a value which  is $\Theta(\IRB(G_X))$ can be done with a simple sweep algorithm. We call an unsatisfied (i.e. not AS) rectangle defined by two points a $+$-type rectangle if its upper point is to the right of its lower point. Let $\IRB^+(G_X)$ to be the point set obtained from $P$ by repeatedly looking at the lowest  $+$ type unsatisfied rectangle (with lowest upper coordinate), and adding a point to the set in the upper-right corner of this rectangle. This process is repeated until no more unsatisfied $+$-type rectangles remain. $\IRB^-(G_X)$ is defined in a symmetric way. See Figure~\ref{geoview} for an illustration of the result of this process. We have shown that $\IRB(P)=\Theta( |\IRB^-(G_X)| +|\IRB^+(G_X)|)$.

The independent rectangle bound is the best known lower bound on $\OPT(X)$, but there are two other bounds that predate it due to Wilber, which are interesting in their own right. They were both introduced at the same time in \cite{DBLP:journals/siamcomp/Wilber89} using the language of BSTs and we briefly state them here in the language of the geometric representation.

\paragraph{Alternation bound.} The alternation bound can be computed using a the geometric view $G_X$ as follows: pick a vertical line $\ell$, and sweep in order of $y$-coordinate through the points of $X$, counting the number of times the points of $G_X$ in this sweep alternate between the left and right side of the line $\ell$. Split the point set $G_X$ into two sets using the line and repeat the process recursively on each side. See Figure~\ref{fig:wi} for an illustration of this process. The lower bound is the total number of alternations in all levels of the recursion. In computing this bound, there is freedom to choose the vertical line at each step; classically the line at each step is chosen to go though the median $x$-coordinate, and we will call this lower bound $\ALT(X)$. While $\OPT(X)=\Omega(\ALT(X))$ it can be shown that it is not always tight; for all $n$ there is a sequence $X_n$ of size $n$ such that $\ALT(X_n)=\Theta(\OPT(X_n) \log \log n)$. Such a sequence can be created randomly by picking $O(\log n)$ nodes that are to the left of the dividing line in all levels of the recursion but one and randomly searching only them. Each random search must take time $O(\log \log n)$ in expectation but will only contribute $O(1)$ to the lower bound calculation. There are several possible avenues for improving this bound: one would be to figure out how to choose the lines best (and \cite{dhthesis} shows that you don't have to restrict the lines to vertical ones), or perhaps change them dynamically. Another avenue for improvement would be to somehow do another type of recursion that would directly reduce the gap exhibited above, possibly reducing it $O(\log \log \log n)$ or $O(\log^* n)$. The main interest in this lower bound is that is it the only bound that has been turned into an algorithm (see Tango trees below); thus improvements on this bound have a reasonable chance of being able to create a better algorithm than what is known.

\begin{figure}
\begin{center}
\includegraphics[width=2in]{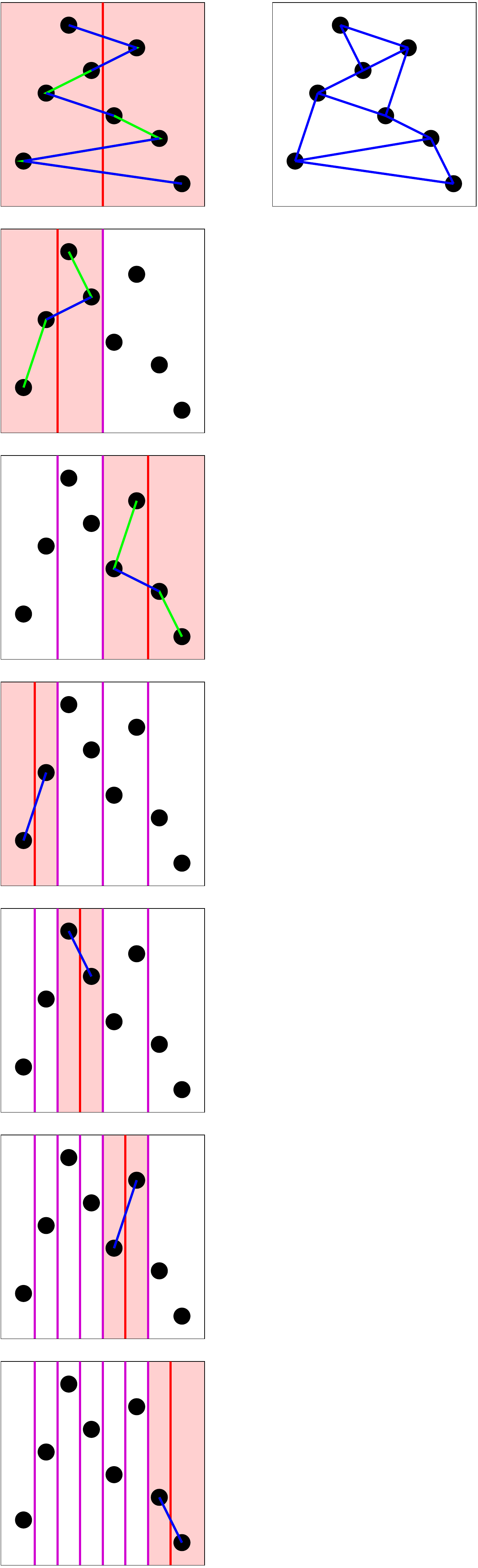}
\end{center}
\caption{The interleave bound. As each line is introduced we restrict out view to the cell bounded by previous lines in and containing the current one. In this cell, we connect the points, top to bottom, and a connection that crosses the current line is colored blue and contributes one unit to the lower bound. Green connections are do not cross the current dividing line and do not contribute to the lower bound. As all blue lines generated are distinct, we can visualize the entire bound using the right diagram.}
\label{fig:wi}
\end{figure}

\paragraph{Funnel bound.}
Given a point $(x_i,i) \in G_X$, the \emph{funnel} of $x_i$, $f(x_i)$ is the set of points below $x_i$ that form unsatisfied rectangles with $(x_i,i)$. For each funnel, look at the points in the funnel sorted by $y$ coordinate, and count the number of alternations from the left to the right of $x_i$ that occur; this is the amortized lower bound for $x_i$; computing and summing this value for all $x_i$ in $X$ gives the lower bound. A different, tree-based view of this bound is as follows: execute $X$ on a binary search tree, and perform a series of single rotations to bring $x_i$ to the root; this BST algorithm was first proposed by Allen and Munro in \cite{DBLP:journals/jacm/AllenM78}. The amortized lower bound for $x_i$ is the number of turns on the path from the root to $x_i$. It is worth noting that this will maintain at each step a treap where the heap value of each $x$ is the last $i$ such that $x_i=x$, or equivalently the working set number of the item. This tree view gives an immediate idea for an algorithm---since only the turns in the tree contribute to the lower bound, is it possible to create a method to maintain a representation of the treap so that an item can be searched in a time proportional to the number of turns in the tree? The main obstacle to this approach is that there could be a linear number of items that are one turn away, thus some kind of amortization would be needed to show that that situation could not happen in each search.

\begin{figure}
\begin{center}
\includegraphics[angle=180,width=1.5in]{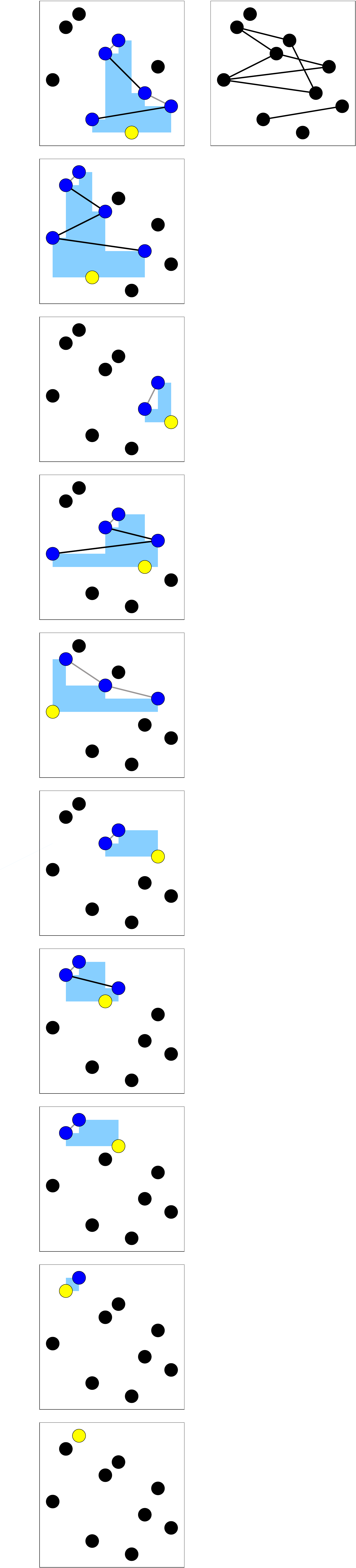}
\end{center}
\caption{The funnel bound. For each yellow node, the blue nodes are the nodes in its funnel, and the lines correspond to pairs in the funnel which cross the yellow node and yield one unit of lower bound. }
\label{fig:wii}
\end{figure}

\paragraph{Relationship among these bounds}

All of these three bounds are lower bounds, that is all of them compute values which are $O(\OPT(X))$.
No relationship is known among the alternation bound and the funnel bound. However, the alternation bound and the funnel bound have been shown to correspond to  sets of independent rectangles, and thus they are asymptotically implied by the independent rectangle bound \cite{DBLP:conf/soda/DemaineHIKP09}. We have mentioned that the alternation bound is not known to be tight. However, while the independent rectangle bound asymptotically implies the funnel bound, the converse is unknown, and we conjecture that they are equivalent.

\paragraph{Improving the bounds.}
In proving the independent rectangle bound, it was shown the need to put a point in each independent rectangle to satisfy the empty rectangles in the original set. Now, adding these points could cause new unsatisfied rectangles, including those that are defined entirely by added points. We call these problems \emph{secondary effects} and it is easy to show that they occur. The question is, are these secondary effects asymptotically significant or not? If one believes that the independent rectangle bound is tight, they could try to show that if one were to put points to satisfy the rectangles in phases, where the points in each phase satisfy the unsatisfied rectangles in the previous phase, the number of problems would form some kind of exponential progression. On the other hand, to show the bound is not tight one would 
need an example where these secondary effects dominate.

\section{Upper bounds}

In this section we survey the various progress towards the dynamic optimality conjecture by producing actual BSTs.

\subsection{Concrete bounds}

One approach has been to come up with concrete closed-form bounds that express the runtime of a particular BST data structure. These bounds initially began with the analysis of splay trees \cite{DBLP:journals/jacm/SleatorT85}, with the working set bound, which says a search is fast if it was searched recently, and the dynamic finger bound \cite{DBLP:journals/siamcomp/Cole00,DBLP:journals/siamcomp/ColeMSS00}  which says that a search is fast if it is close in key value to the previous search. In \cite{DBLP:journals/tcs/BadoiuCDI07}, we proposed combining these bounds into one which bounds a search as being fast if it is close in key value to some search that has been searched frequently; a BST-model algorithm with this bound was claimed in \cite{jdthesis} but may be buggy \cite{sleatortalk}. These bounds all can easily be seen to not be tight, that is there are classes of sequences $X$ such that they are $\omega(\OPT(X))$. Can refining bounds of this type have any hope of a closed form solution that is an asymptotically tight expression of $\OPT(X)$? For example, in the closely related problem of the runtime of the best static tree where each search beings where the previous one ended, a closed-form expression for the asymptotic runtime was obtained \cite{DBLP:journals/corr/abs-1304-6897}. However, for optimal BST's with rotations, this approach seems difficult as it is easy to come up with search sequences which can be executed quickly on a BST but which all known concrete upper bounds are not tight and which seem to defy a simple formulaic characterization.

%

\subsection{Tango trees \cite{DBLP:journals/siamcomp/DemaineHIP07}}
In the language of competitive analysis, the problem of finding a dynamically optimal binary search tree is to find one which is $O(1)$-competitive with the best offline tree. 
Any tree with $O(\log n)$ search time is trivially $O(\log n)$ competitive. Tango trees are a data structure that was created to have a non-trivial approximation factor of $O(\log \log n)$.
Specifically, they are created to be within a $O(\log \log n)$ factor of the alternation lower bound.

We present another view of computing the alternation bound, one which leads easily to the central idea of the Tango tree construction. This computation is presented algorithmically. To compute the alternation bound, envision a \emph{reference tree} which is a balanced binary containing $[1..n]$. Each nonleaf node in the tree has one of its children designed as the preferred child---the preferred child is designated based on which subtree of that node has had the most recent search. Now to compute the bound, go though the sequence $X$ and execute each search on the reference tree. The process of executing a search involves starting at the root and following children, which are either preferred or not; at the end of the search the nodes where the search did not follow the preferred child are updated to reflect that the preferred child has changed and is now aligned with the search path; the sum of these preferred child changes is equivalent to the alternation lower bound.
Given a node, call the preferred path the path in the reference tree obtained by following preferred children starting from that node until a leaf is reached. Due to the balanced nature of the reference tree, the size of any preferred path is $O(\log n)$. Given this setup, the idea behind the Tango tree is to store each preferred path in a balanced binary search tree of height $O(\log \log n)$. Thus a search in the reference tree that involved following $\log n$ nodes among $k$ preferred paths (and this has a lower bound of $O(k)$) can be executed in time $O((k+1)\log \log n)$ time.
Details of the Tango tree involve arranging the BST's created from each preferred path into one BST, and using split and merge operations to maintain the changing of the preferred paths.
However, this method is limited by the fact that the alternation bound is sometimes tight and sometimes off by a $\Theta(\log \log n)$ factor. Thus, no improvement in the competitive ratio is possible while still using the alternation bound as-is as the basis of the competitive ratio. Note that \cite{DBLP:conf/wads/DerryberryS09} improved Tango trees to have some additional desirable properties, such as $O(\log n)$ worst-case time per search, as opposed to $O(\log n \log \log n)$ in their original presentation.

\subsection{Greedy}
If one were to come up with an idea for candidates for the best possible offline method there is one greedy method that stands out. Search for the current item. Then for (asymptotically) free one can rearrange the nodes on the search path into any tree. The heuristic that makes the most sense would be to place the node to be searched next at the root, or if the node to be searched next is not on the search path, place the subtree that contains it as close to the root at possible. Then, recurse on the the remaining nodes on the search path.  This method was first proposed by Lucas \cite{jltr}.

In the geometric view, there is a natural greedy method to find an ASS superset of $G_X$: from bottom to top, add points on every row so that the point set is satisfied from that row down. See Figure~\ref{geoview} for an illustration of this process.
It turns out that by applying the geometric-to-BST conversion to this method, Lucas's greedy tree method is obtained; thus the two greedy methods are in fact identical.

While this method seems intuitively to be a good idea, basic facts like $O(\log n)$ amortized time per search were not known until the work of Fox \cite{DBLP:conf/wads/Fox11}, who also showed that there is an equivalent online BST to this method. Showing that this method which greedily looks into the future has an equivalent online method provides some support for the belief that the best online and offline BST algorithms have asymptotically the same runtime.

In the geometric view, recall that the independent rectangle lower bound can be computed by sweeping twice though $R_X$ and satisfying the $+$ and $-$ unsatisfied rectangles separately. The greedy method is a single sweep though $R_X$ satisfying both the $+$ and the $-$ rectangles at the same time (again, see Figure~\ref{geoview}). Proving that the point sets obtained though these two methods are within a constant factor of each other would be enough to show the greedy method is a dynamically optimal binary search tree. We have spent much time coding and looking for ways to prove such a correspondence to no avail.

\begin{figure}
\begin{center}
\includegraphics[trim=80 200 100 170, angle=180, width=4.5in]{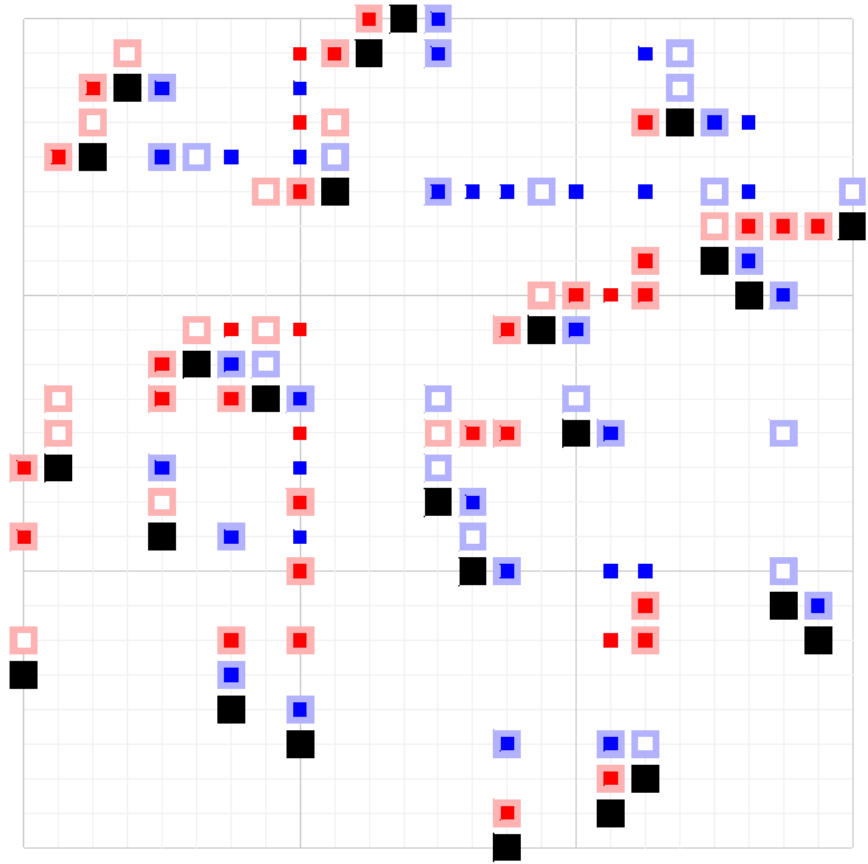}
\end{center}
\caption{The geometric view of binary search trees.
A black dot at location $(x,y)$ represents that we wish to search for key value $x$ at time $y$; it is set set $G_X$. The black dots, combined with the solid blue and red represent an execution of Lucas's greedy future algorithm to execute this search sequence; a dot at $(x,y)$ represents the greedy algorithm touching item with key $x$ at time $y$. The solid dots are \emph{satisfied}, that is for every two dots not in the same row or column, you can find a third one on the rectangle they define. Color simply represents being to the left or to the right of black. The $\Box$-shaped markers are a visual representation of a the incremental construction of the independent rectangle lower bound for the minimal satisfied superset of the black points. Showing Lucas's method is dynamically optimal is equivalent to showing the solid and $\Box$'s are always within a constant factor of each other for any such diagram.} 
\label{geoview}
\end{figure}

\subsection{Combining trees}
In \cite{DBLP:journals/corr/abs-1304-7604}, we have shown that given any constant number of online BST algorithms (subject to certain technical restrictions described in the paper), there is an online BST algorithm that performs asymptotically as well on any sequence as the best input BST algorithm. If the output algorithm did not have to be in the BST model this would be trivial as the input algorithms could just be run in parallel; however the BST-model restriction makes this non-trivial. It is open whether or not it is possible to combine a superconstant number of BST algorithms. This would be of interest as a dynamically optimal BST could be viewed as combining \emph{all} algorithms and taking the minimum. As the number of BST algorithms can be limited to be a function of $n$ (see next section), this opens the possibility of having an algorithm with a runtime of $O(\OPT(X)+f(n))$, for some possibly huge function $n$.


\subsection{Search optimality.} In \cite{DBLP:journals/algorithmica/BlumCK03}, the notion of \emph{search optimality} was presented. The \emph{search cost} of a search BST algorithm is simply the depth of the node to be searched. Any rotations or pointer movements off the search path are free; in this way the BST can be arbitrarily reconfigured between searches at zero cost. It was shown that there is a BST model algorithm for which the search cost is $O(\OPT(X))$. The general method was to use a machine learning approach to finding the best tree after each search based on the searches performed so far. If one were to try to adapt this method to the standard online BST model cost function, a reasonable starting point would be to try to determine if there is any cohesion of the trees produced by the method from one search to the next, and to try to figure out if one could use such cohesion to transform one tree to the next in time proportional the the search cost. 

\section{Online optimality}

In this section we present the only new result of this paper: we present the best possible online BST algorithm for sufficiently long search sequences. In particular, we prove the following:

\begin{theorem} \label{tha}
There is an online BST algorithm $OnOpt$ such that if there is an online algorithm $A$ such that $R_A(X)=O(\OPT(X)+f(n))$ for some function $f(n)$, then there is an algorithm $OnOpt$ and a function $g(n)$ such that $R_{OnOpt}(X)=O(\OPT(X)+g(n))$.
\end{theorem}

 The algorithm is decidedly not in the real-world BST model, and is a relatively straightforward application of known methods from learning theory.

We begin by summarizing a classic result from learning theory (see \cite{DBLP:journals/toc/AroraHK12} for a survey of its origins, variants and applications). The setup is that there are a sequence of events $Z=z_1,z_2, \ldots, z_\ell$ which are presented in an online manner---each event is an integer in the range $[1..\rho]$. Before each event is revealed, one of $\eta$ \emph{experts} numbered $[1..\eta]$ is chosen. After the event is chosen a \emph{penalty} is determined based on an $ \eta \times \rho$ table $M$ which assigns penalties to each combination of event and expert; $M[a,z]$ is the penalty if expert $a$ was chosen and event $z$ happened.

Thus, if a single expert $a$ were to be chosen for all events, the total penalty would be 
$\sum_{k=1}^{\ell} M[a,z_k]$. The main result we will use is that it is possible to pick online an expert before each event such that the total penalty is asymptotically that of the best expert:

\begin{theorem}[Weighted Majority Algorithm]
For any $\epsilon>0$, there is an online choice of expert $E=e_1, e_2, \ldots z_\ell$ such that

$$ \sum_{k=1}^m M[e_k,z_k] \leq  \frac{\rho \ln \eta}{\epsilon} + (1+\epsilon) \min_a \sum_{k=1}^m M[a,z_k] $$
\end{theorem}

Now we apply this theorem to BSTs.
Let $X=x_1, x_2, \ldots x_m$ be a sequence of searches in a BST-model data structure containing the integers $1..n$; for convenience we assume $m$ is a multiple of $f(n)$, and we assume $f(n)\geq n$. 
We let the events be the $n^{f(n)}$ different search sequences of length $f(n)$; thus $\rho=n^{f(n)}$.
We divide the search sequence into \textit{epochs} of size $f(n)$, and denote the $i$th epoch as $z_i$. Note that each epoch $z_i$ is an event.

How many BST-model algorithms are there to execute an epoch, assuming at the beginning and end of each epoch the BST is in a canonical state (e.g.~a left path)? There are at most 
$n^{f(n)} 5^{O(nf(n))}$. This is because you can encode an algorithm by encoding the  $O(n f(n))$ fundamental operations spent to execute each of the $n^{f(n)}$ possible epochs, and the 5 possible BST unit-cost operations at unit of time executing each epoch. We view the set of online BST epoch algorithms as the experts. Thus $e \leq n^{f(n)} 5^{O(nf(n))}$. 
This is a gross overestimate as this counts the offline algorithms, and does not cull those which do not properly execute each search. The cost $M[a,z_k]$ is simply the runtime of BST epoch algorithm $a$ on epoch $k$.

Plugging this into Theorem~\label{main} gives:

\begin{lemma} \label{bsta} There is a way to choose an algorithm $A_k$ at each epoch such that:
$$ \sum_{k=1}^m M[A_k,x_k] = O\left( n^{f(n)}n f(n)  + \min_a \sum_{k=1}^m M[a,x_k] \right) $$
\end{lemma}

Now, recall  $\OPT(X)$ is the fastest any offline BST-model algorithm can execute the search sequence $X$.

\begin{lemma}
Given a search sequence $X$ of length $m$ on a set of size $n$, let $z_i$ be a search sequence of size $n$ which is the $i$th epoch of $S$. Then for any $f(n)\geq n$
$$\OPT(S)=\Theta\left( \sum_{i=1}^{m/f(n)} (\OPT(z_i)+f(n)) \right)$$
\end{lemma}

\begin{proof}
Follows directly from the fact that any BST can be converted into any other in $O(n)$ time. Thus you can be forced into a canonical state every $n$ searches and this only changes the optimal time by a constant.
\end{proof}

These lemmas give a proof of Theorem~\ref{tha}. Specifically, if there is an unknown online BST algorithm with runtime $O(\OPT(X)+f(n))$, then using the weighted majority algorithm to pick an algorithm to run every $f(n)$ steps yields an online BST algorithm that runs in time $O(\OPT(X)+n^{f(n)}n f(n))$, which is $O(\OPT(X))$ for sufficiently long sequences $X$.


\bibliography{bib,extrabibs}

\end{document}